\documentclass[11pt]{article}
\usepackage{fullpage}

\usepackage{latexsym}
\usepackage[caption=false,font=normalsize,labelfont=sf,textfont=sf]{subfig}
\usepackage{cite}
\usepackage[cmex10]{amsmath}
\usepackage{bbm}
\usepackage{amssymb}
\usepackage{graphicx}
\usepackage{algorithm}
\usepackage{algpseudocode}
\usepackage{array}
\usepackage{flexisym}
\usepackage[mathscr]{euscript}
\usepackage{mathtools}
\usepackage{amsthm}
\usepackage{url}
\usepackage{siunitx}
\DeclareGraphicsExtensions{.density}

\DeclarePairedDelimiter\ceil{\lceil}{\rceil}
\DeclarePairedDelimiter\floor{\lfloor}{\rfloor}

\newtheorem{theorem}{\it{\bf Theorem}}

\newtheorem{lemma}{\bf Lemma}

\let\originalleft\left
\let\originalright\right
\renewcommand{\left}{\mathopen{}\mathclose\bgroup\originalleft}
\renewcommand{\right}{\aftergroup\egroup\originalright}

\DeclareMathOperator{\E}{\mathbb{E}}
\newcommand{\abl }{\left|}
\newcommand{\abr }{\right|}
\newcommand{\sql }{\left[}
\newcommand{\sqr }{\right]}
\newcommand{\crl}{\left\{ }
\newcommand{\crr}{\right\} }
\newcommand{\prl}{\left(}
\newcommand{\prr}{\right)}
\newcommand{\eps}{\varepsilon}

\newcommand{\argmin}{\operatornamewithlimits{arg\ min}}

\newcommand{\cmin}{C_\mathrm{min}}
\newcommand{\thetainv}{\theta^{-1}}
\newcommand{\Zt}{Z_{\theta}}

\newcommand{\nn}{\nonumber}
\newcommand{\defi}{\triangleq}

\newcommand{\MSE}{\mathrm{MSE}}
\newcommand{\MSEo}{\mathrm{MSE}_{\mathrm{o}}}
\newcommand{\vX}{\mathbf{X}}
\newcommand{\vw}{\mathbf{w}}
\newcommand{\vwo}{\mathbf{w}_{\mathrm{o}}}
\newcommand{\vU}{\mathbf{U}_{\scriptstyle \vtheta}}
\newcommand{\vOne}{\mathbf{1}}
\newcommand{\vZ}{\mathbf{Z}_{\scriptstyle \vtheta}}
\newcommand{\vtheta}{\boldsymbol{\theta}}

\newcommand{\Yh}{\hat{Y}}
\newcommand{\C}{C}
\newcommand{\G}{G}
\newcommand{\D}{D}

\newcommand{\Not}{N_{\mathrm{o}}\prl \tau \prr}
\newcommand{\aot}{a_{\mathrm{o}}\prl \tau \prr}

\newcommand{\Vt}{V \prl \tau \prr}
\newcommand{\cost}{\mathrm{Cost}}
\newcommand{\costt}{\mathrm{Cost}_{\tau}}
\newcommand{\tinv}{\tau^{-1}}
\newcommand{\kappat}{\kappa_{\tau}}

\begin{document}
\title{Cost-Performance Tradeoffs in Fusing Unreliable Computational Units
\thanks{This work was supported in part by Systems on Nanoscale Information fabriCs (SONIC), one of the six SRC STARnet Centers, sponsored by MARCO and DARPA, and in part by the Center for Science of Information (CSoI), an NSF Science and Technology Center, under
grant agreement CCF-0939370.
Some of the results in this paper appeared in an early version at a conference \cite{DonmezRSV2016}.}}
\author{Mehmet A.~Donmez\thanks{donmez2@illinois.edu} \and Maxim Raginsky\thanks{maxim@illinois.edu} \and Andrew C.~Singer\thanks{acsinger@illinois.edu}
\and Lav~R.~Varshney\thanks{varshney@illinois.edu}}
\maketitle
\begin{abstract}
We investigate fusing several unreliable computational units
that perform the same task.
We model an unreliable computational outcome as an additive perturbation to its error-free result in terms of its fidelity and cost.
We analyze performance of repetition-based strategies
that distribute cost across several unreliable units and fuse their outcomes.
When the cost is a convex function of fidelity,
the optimal repetition-based strategy in terms of incurred cost while achieving a target
mean-square error (MSE) performance may fuse several computational units.
For concave and linear costs, a single more reliable unit incurs lower cost
compared to fusion of several lower cost and less reliable units
while achieving the same MSE performance.
We show how our results give insight into problems from theoretical neuroscience, circuits, and crowdsourcing.
\end{abstract}

\section{Introduction}
We consider the problem of fusing outcomes of several unreliable computational units that perform the same computation under cost and fidelity constraints.
We formalize the relationship between the fidelity of each unit and the cost
associated with it, and explore this tradeoff in a number of practical problems.
Consider, for instance, the capacity of an additive white Gaussian noise (AWGN) channel, which is a
logarithmic function of the signal-to-noise (SNR) ratio.
In this scenario, the capacity can be increased at the expense of requiring a higher SNR, which introduces a tradeoff between cost (SNR) and performance (rate).
Note also that the Fisher information in estimation is often a linear function of SNR, leading to a different cost-performance tradeoff \cite{Kay2010}.

Building reliable systems out of unreliable components has attracted substantial interest in
circuits and systems \cite{Neumann1956,Tryon1962,WinogradC1963},
information theory \cite{Pippenger1988, HajekW1991, EvansP1998}, and
signal processing \cite{YangGK2016}.
In \cite{Neumann1956}, Von Neumann investigated error in logic circuits from a statistical point of view and demonstrated that repeated computations followed by majority logic may yield reliable
results even when the underlying components are unreliable.
In  \cite{Tryon1962}, Tryon introduced a technique called quadded logic,
which corrects errors by a redundant design of logic gates.
Moreover, the authors of \cite{Pippenger1988, HajekW1991, EvansP1998} investigated
reliable computation by formulas in the presence of noise.
More recently, the authors of \cite{YangGK2016} considered energy-reliability tradeoffs in computing linear transforms implemented on unreliable components.

Fusion of the outputs collected from several sensors has been considered
in distributed detection, estimation, classification, and optimization in sensor networks \cite{ViswanathanV1997, IshwarPRP2005, RiberioG2006_1, RiberioG2006_2, BarbarossaSL2013,AldosariM2004,MarcoN2004}.
Often, spatially distributed sensors locally perform a decision-making task and
send their outputs, under bandwidth constraints, to a fusion center that forms a final decision.
In most practical applications, these sensors are battery-powered devices with limited accuracy
and computational capabilities,
so their performance is critically affected by the resources allocated to them,
introducing a cost-performance tradeoff.
The authors of \cite{ AldosariM2004} studied tradeoffs between the number of sensors,
resolution of quantization at each sensor, and SNR.
Similarly, \cite{MarcoN2004} considered the tradeoff between reliability and efficiency in distributed source coding for field-gathering sensor networks.
In general, the main goal is to make a reliable final decision in a cost-efficient manner based on these
unreliable sensors subject to resource and reliability constraints.

A fundamental question that arises in fusing several unreliable computational units is
how a limited budget should be allocated across several unreliable units,
where adding a new unit incurs a baseline cost as well as an incremental cost,
and also increases the cost of fusion.
That is, what is the {\it optimal} approach in terms of cost-performance tradeoff?
Although existing work in fault-tolerant computing and in-sensor networks
focus on different pieces of this problem, a more general treatment that jointly considers
cost and performance is necessary.
This paper is an attempt to combine insights from both fields into a unified framework
that captures characteristics of a range of problems.
In particular, we show how our framework and results are connected to problems from
neuroscience, circuits, and crowdsourcing in Section~\ref{sec:app}.

In this paper, we present an abstract framework to explore the fundamental tradeoff
between cost and performance achievable through forms of redundancy.
We model unreliability in any computational unit as an additive random perturbation,
where the variance of the perturbation is inversely related to its fidelity.
We cast the main task as inference of the error-free computation based
on noisy computational outcomes.
Each computational unit incurs some cost, which is a function of its fidelity,
that includes a baseline cost incurred simply to operate the unit.

We define a class of repetition-based strategies, where each strategy distributes
the total cost across several unreliable computational units and fuses their outputs.
We note that the fusion operation also incurs some cost, which is a function of the number of
individual computational units to be fused.
We measure the inference performance of each strategy in terms of MSE
between its final output and the error-free computation.

We consider optimal repetition-based strategies under convex, linear, and concave cost functions rather than restricting to specific cost functions.
For convex costs, there are two main cases.
In the first case, we prove that using only a single and more reliable computational outcome is more cost-efficient than the fusion of several lower cost but less reliable computational outcomes.
In the second case, however, we demonstrate that the optimal strategy uses several
computational outcomes instead of a single more reliable one.
Intuitively, the convexity of the cost function disperses the cost across
several less reliable computational outcomes with smaller individual costs.
For linear or concave costs, the optimal strategy is to use a single and more reliable computational outcome.

\section{Problem Description}
\label{sec:problem}
Consider the problem of fusing outcomes of several unreliable computational
units subject to cost and fidelity constraints. 
We first introduce a model of an unreliable
computational outcome as an additive perturbation to its error-free result in terms of its fidelity and cost.
We next consider a class of repetition-based strategies that distribute cost across
several parallel unreliable units and fuse their outcomes to
produce a final estimate of the error-free computation.

Suppose a vector of input signals $\vX = \prl X_1, \ldots, X_k \prr $ is processed
to yield the error-free computation,
\begin{align*}
	Y = f\prl \vX \prr,
\end{align*}
where $f\prl \cdot \prr$ is some arbitrary target function.
Instead, we observe an unreliable computational outcome,
\begin{align*}
	\Zt = Y + U_{\theta},
\end{align*}
where $U_{\theta}$ is a zero-mean perturbation with variance $\thetainv$.
Here, $\theta$ is the fidelity of the unreliable computational outcome $\Zt$.
We assume that $Y$ and $U_{\theta}$ are uncorrelated,
that is, $\E \sql Y U_{\theta} \sqr = \E \sql Y \sqr \E \sql U_{\theta} \sqr$ holds,
whether or not $Y$ is a random variable.

By Chebyshev's inequality, the unreliable outcome $\Zt$ with fidelity
$\theta>0$ satisfies, for any $\eps>0$,
\begin{align}
	\Pr\prl \abl \Zt - Y \abr\geq\varepsilon\prr \leq \frac{1}{\varepsilon^2\theta}.
	\label{eq:chebyshev}
\end{align}
This implies the unreliable outcome $\Zt$ converges to the error-free computation in probability as the fidelity tends to infinity.
However, as the fidelity parameter $\theta$ increases, the cost $\C(\theta)$ incurred to guarantee that level of fidelity also increases, introducing a {\it cost-fidelity tradeoff}.
Note that this holds both when $X_i$ for $i=1,\ldots,k$, or $Y$, are random as well as when
they are purely deterministic.

In this model, we must incur a cost $\C \prl \theta \prr$ to get
the unreliable outcome $\Zt$ with fidelity $\theta>0$,
which we assume to be a strictly increasing function of $\theta$.
In particular, we assume
\begin{align*}
	\C \prl \theta \prr = \cmin + \G \prl \theta \prr,
\end{align*}
where
$
	\cmin \triangleq \inf_{\theta>0} \C \prl \theta\prr \geq 0
$
is the minimum (baseline) cost, and $\G \prl \theta\prr$ is an increasing and twice differentiable
incremental cost function with $G(0)=0$.
In the sequel, we focus on three classes of cost functions: convex, linear, and concave.

We define a class of repetition-based strategies that fuse
the outputs of several computational units to estimate $Y$.
For any positive integer $N$, a repetition-based strategy $S_{N}$,
with weights
$
	\vw = \prl w_1,\ldots,w_N \prr \in \mathbb{R}^N
$
and fidelities 
$
	\vtheta = \prl \theta_{1},\ldots,\theta_{N}\prr\in\prl 0,\infty\prr^N,
$
linearly combines the outcomes of $N$ parallel unreliable units with fidelities
$\vtheta$ using the weights $\vw$.
That is, if we denote each unreliable outcome with
the fidelity $\theta_i$ and the cost $\C\prl \theta_i \prr$ as
\begin{align*}
	Z_{\theta_i} = Y + U_{\theta_{i}} ,
\end{align*}
for $i=1,\ldots,N$, then the final output of this strategy $S_N$ is
\begin{align}
	\Yh_N \prl \vw; \vtheta \prr 
	&\triangleq \vw^T \vZ
	= Y \prl \vw^T\vOne \prr + \vw^T \vU \label{eq:cont_single}
\end{align}
where $\vZ \defi \prl Z_{\theta_{1}},\ldots, Z_{\theta_{N}}\prr$,
$\vU \defi \prl U_{\theta_{1}},\ldots, U_{\theta_{N}}\prr$, and $\vOne = \prl 1,\ldots,1 \prr\in\mathbb{R}^N$ is a vector of ones.
In particular, we assume that $U_{\theta_i}$s are uncorrelated with each other.

The cost incurred by the strategy $S_N$ with fidelities $\vtheta$ is
\begin{align}
	\sum_{i=1}^N \C \prl \theta_{i} \prr+ \D \prl N\prr, \nn
\end{align} 
where $\D \prl N\prr$ is the fusion cost, i.e., the cost of linear combination.
We assume that the function
$
	\D :\mathbb{Z}_+\rightarrow\mathbb{R}_+ 
$
is increasing, as fusing a larger number of computational units has higher cost than fewer.
Note that the fusion cost is super-linear in $N$ in that it
requires at least $O\prl N \prr$ multiplications and additions.
In particular, we assume that $\D\prl N \prr$ is convex in $N$.

\section{Performance Analysis}
\label{sec:perf}
Here, we consider the MSE performance of each repetition-based strategy
in estimating the error-free computation $Y$.
For any positive integer $N$, the strategy $S_N$ with a weight vector $\vw\in\mathbb{R}^N$
and a fidelity vector $\vtheta\in\prl 0,\infty\prr^N$ achieves the MSE
\begin{align}
	\MSE \prl \vw, \vtheta \prr \defi \E\sql \prl \Yh_N \prl \vw; \vtheta \prr - Y \prr^2\sqr. \label{eq:MSEdefined}
\end{align}
In particular, we derive the minimum MSE (MMSE) achievable by
this strategy $S_N$ while producing an unbiased output:
\begin{align}
	\MSEo \prl \vtheta \prr
	\defi \min\limits_{\vw^T \vOne = 1} \MSE \prl \vw, \vtheta \prr, \nn
\end{align}
where $\vwo$ is the corresponding minimizer.

\begin{lemma}
\label{lemma:compare}
Suppose that for any positive integer $N$, the strategy $S_N$ fuses the outcomes of
$N$ parallel computational units with fidelities
$\vtheta \in \prl 0,\infty\prr^N$.
Then the MMSE achievable by this strategy $S_N$ while producing an unbiased estimate
of $Y$, and the corresponding weights are
\begin{align}
	\MSEo \prl \vtheta \prr
	= \frac{1}{\vtheta^T\vOne}, \;\;
	\vwo
	= \frac{\vtheta}{\vtheta^T\vOne}, \label{eq:optimalWeights}
\end{align}
respectively.
\end{lemma}

\begin{proof}
We provide the proof in Appendix~\ref{app:compare}.
\end{proof}

Thus, Lemma~\ref{lemma:compare} provides the strategy $S_N$ achieving the MMSE
for a given fidelity vector $\vtheta \in \prl 0,\infty \prr^N$.
For any positive integer $N$, whenever we refer to the strategy $S_N$,
we use the optimal weights given in \eqref{eq:optimalWeights}, so that its output is
\begin{align*}
	\Yh_N \prl \vwo; \vtheta \prr = \vwo^T \vZ = \frac{\vtheta^T \vZ}{\vtheta^T\vOne}.
\end{align*}
We next study a particular scenario, where $U_{\theta}$ is sub-Gaussian.

\subsection{Sub-Gaussian Perturbations}
Here, we consider a case where the perturbation $U_{\theta}$ is sub-Gaussian with parameter $\thetainv$, which means \cite{BoucheronLM2013}
\begin{align}
	\E\sql e^{\lambda U_{\theta}  }\sqr \leq \exp \prl \frac{\lambda^2}{2\theta}\prr,\;\;
	\forall \lambda\in \mathbb{R},
	\label{eq:subGaussianModel}
\end{align}
or equivalently, the probability of absolute deviation of $\Zt$ from $Y$ satisfies, for any $\eps>0$,
\begin{align}
	\Pr\prl \abl \Zt - Y \abr \geq \eps \prr \leq 2\exp\prl -{\eps^2\theta/2} \prr.
	\label{eq:tailBound}
\end{align}
The tail bound in \eqref{eq:tailBound} decreases faster (with increasing $\theta$)
than the bound in \eqref{eq:chebyshev}.
Sub-Gaussian distributions can be used to model a wide range of
stochastic phenomena including
Gaussian and uniform distributions, or distributions with finite or bounded support.
Note that a weighted sum of finitely many 
sub-Gaussian random variables is also sub-Gaussian \cite{BoucheronLM2013}.
By appling this result to the output of a strategy $S_N$ with $\vw\in\mathbb{R}^N$
and $\vtheta\in\prl 0,\infty\prr^N$, we get, for any $\eps>0$,
\begin{align*}
	\Pr\prl \abl \Yh_N \prl \vw; \vtheta \prr - Y \abr \geq \eps \prr 
	\leq 2\exp\prl - \frac{\eps^2 }{ \sum_{i=1}^N w_i^2/\theta_i} \prr. \nn
\end{align*}
The weights minimizing the upper bound under $\vw^T\vOne = 1$, and the resulting bound are
known to be
$
	\vwo = \vtheta/ \vtheta^T\vOne,
$
and
\begin{align}
	\Pr\prl \abl \Yh_N \prl \vwo; \vtheta \prr - Y \abr \geq \eps \prr 
	\leq 2\exp\prl -{\eps^2 \vtheta^T\vOne/2} \prr,\nn
\end{align}
for any $\eps>0$, respectively.

We emphasize that, in this case, even though the performance is measured in terms of probability
of absolute deviation from the error-free computation,
the optimal weights are exactly the same as the ones minimizing the MSE.
Hence, same results apply to both cases when comparing the cost-performance tradeoff
of the repetition-based strategies.

In this section, we analyzed the MSE performance of repetition-based strategies.
More precisely, for any positive integer $N$ and a fidelity vector $\vtheta\in\prl 0,\infty\prr^N$,
we derived the optimal weights for the strategy $S_N$ in terms of minimizing the MSE.
Based on these results, we next investigate the cost-performance tradeoff
for the class of repetition-based strategies under the classes of convex, linear, and concave
cost functions.

\section{Cost-Performance Tradeoff}
\label{sec:cost}
We investigate the performance of repetition-based strategies
under convex, linear, and concave cost functions in terms of the tradeoff between the total incurred cost and the final MSE performance in estimating the error-free computation.

We first analyze the case where the cost $\C \prl \theta \prr$ is a convex function of the fidelity
$\theta$.
We characterize the optimal strategy, based on the desired MSE performance as well as the baseline and fusion cost functions.
In particular, we show that the optimal cost-performance tradeoff may be achieved by some strategy $S_N$ with $N>1$ under certain conditions.

We next consider the case where the cost $\C \prl \theta \prr$ is a linear function of the fidelity parameter $\theta$, and
show that strategy $S_1$ is optimal among repetition-based strategies.
We finally study the concave cost scenario, and demonstrate results similar to the linear cost function case.

To compare cost-performance tradeoffs of repetition-based strategies, we constrain each strategy to guarantee the same MSE performance.
More precisely, given some $\tau > 0$, we assume that the strategy $S_N$ with
$\vtheta\in\prl 0,\infty\prr^N$ satisfies
\begin{align*}
	\tau = \MSEo \prl \vtheta \prr = \frac{1}{\vtheta^T \vOne},
\end{align*}
or equivalently, $ \tinv = \vtheta^T\vOne$, for any positive integer $N$.
We also define the total cost incurred by this strategy $S_N$, which achieves
$\MSEo \prl \vtheta \prr = \tau$, as
\begin{align*}
	\costt \prl N \prr \defi \sum_{i=1}^N C \prl \theta_{i} \prr+ D \prl N\prr. \nn
\end{align*}

\subsection{Convex Cost Functions}
We study the cost-performance tradeoff for the class of repetition-based strategies
under a convex cost function.
This case corresponds to a {\it law of diminishing returns} between cost and fidelity,
which may drive the dispersion of cost across several less reliable computational units
with smaller individual costs.
We show that there are two main cases, where in the first case some strategy $S_N$ with $N>1$ may incur the minimum total cost achievable by the repetition-based strategies while achieving the same MSE, whereas in the second case the strategy $S_1$ is optimal in terms of cost-performance tradeoff,
i.e., no repetition or fusion is required.

Consider a uniform fidelity distribution across several unreliable computational outcomes, given by
\begin{align}
	\theta_{i} \triangleq  \frac{1}{\tau N},\;\;i=1,\ldots,N,
	\label{eq:uniform}
\end{align}
which implies that the constraint $\MSEo \prl \vtheta \prr = \tau$ is satisfied.
In fact, the following lemma shows that the optimal fidelity distribution satisfying the MSE constraint
in terms of minimizing the total cost is in fact uniform.
\begin{lemma}
\label{lemma:uniform}
For any $\tau>0$, the uniform fidelity distribution given by \eqref{eq:uniform} is the unique solution
to the optimization problem:
\begin{align*}
	\min_{\vtheta \in \mathbb{R}^N_+} \sum_{i=1}^{N} C\prl \theta_{i}\prr
\end{align*}
subject to $\vtheta^T\vOne = \tinv$
when the cost function $C\prl \theta \prr$ is convex.
\end{lemma}
\begin{proof}
The proof is given in Appendix~\ref{app:uniform}.
\end{proof}

Hence, we only consider the case where the strategy $S_N$, for each positive integer $N$,
uses the fidelities in \eqref{eq:uniform}.
The total cost incurred by this strategy $S_N$ is:
\begin{align}
	\costt \prl N \prr
	&= \sum_{i=1}^N C \prl \frac{1}{\tau N}\prr + D\prl N \prr \nn \\
	&= N G \prl \frac{1}{\tau N} \prr + N\cmin + D\prl N \prr. \label{eq:ConvexCost}
\end{align}

To investigate the behavior of the total cost, we define its continuous relaxation as
\begin{align}
	\costt &: \sql 1, \infty \prr \rightarrow \prl 0,\infty \prr \nn\\
	\costt \prl a \prr
	&\defi a G \prl \frac{1}{\tau a} \prr + a\cmin + D\prl a \prr,\nn
\end{align}
where $D\prl a \prr$ is a twice differentiable continuous relaxation of the fusion cost function
$D\prl N \prr$.
We first demonstrate that $\costt \prl a \prr$ is a convex function in $a$.
\begin{lemma}
\label{lemma:convexTotal}
The total cost function $\costt \prl a \prr$ is convex in $a$.
\end{lemma}
\begin{proof}
The proof is provided in Appendix~\ref{app:convexTotal}.
\end{proof}

Convexity of $\costt \prl a \prr$ implies that it has a unique minimizer on any given compact subset of
its domain $\sql 1, \infty\prr$.
In particular, note that
$\costt \prl 1 \prr = G \prl \tinv \prr + \cmin$, and
$ \costt \prl a \prr \rightarrow \infty$ as $a\rightarrow\infty$.
Therefore, the total cost function $\costt \prl a \prr$ has a unique
and finite minimizer $ \aot \in \sql 1,\infty \prr$.
Also, there exists a corresponding unique optimal repetition-based strategy, which we denote as the
strategy $S_{\Not}$ where
\begin{align}
	\Not = \argmin\limits_{N\in\crl \floor{\aot},\ceil{\aot}\crr} \costt \prl N \prr.
	\label{eq:optimalN}
\end{align}
is a finite positive integer (a function of $\tau$), that minimizes the total incurred cost while achieving the desired MSE of $\tau>0$.

We next characterize conditions under which the optimal repetition-based strategy
either uses a single but more reliable computational unit, that is, $\Not=1$,
or distributes the cost across several unreliable computational units and fuses their outcomes,
that is, $\Not>1$.
In the latter case, we implicitly derive the optimal strategy as a function of the desired MSE level
$\tau$, the baseline cost $\cmin$, and the fusion cost function $D\prl \cdot \prr$.
The next theorem characterizes these cases in terms of the first derivative of the
fusion cost and the baseline cost.

\begin{theorem}
\label{thm:ultimateCondition}
For any given $\tau>0$, the minimizer of $\costt \prl a\prr$ satisfies $\aot>1$ if and only if
\begin{align*}
	\cmin + D^{\prime} \prl 1 \prr <  \Vt
\end{align*}
where
\begin{align}
	\Vt \defi \tinv G^{\prime} \prl \tinv \prr - G \prl \tinv \prr. \label{eq:V}
\end{align}
\end{theorem}
\begin{proof}
We define
$
	\kappat \prl a \prr \defi \partial \costt\prl a\prr / \partial a,
$
and observe that from Lemma~\ref{lemma:convexTotal}, $\kappat \prl a \prr$ is nondecreasing
and continuous in $a$ since $\costt\prl a\prr$ is a twice differentiable and convex function of $a$.
Hence, whenever $\kappat \prl 1\prr \geq 0$, we have $\kappat \prl a\prr \geq 0$ for any $a>1$.
It implies that $\costt\prl a\prr$ is a nondecreasing function of $a$ on $\sql 1,\infty\prr$,
and minimized at $\aot =1$.
When $\kappat \prl 1\prr<0$, $\costt\prl a\prr$ is minimized
at some finite $\aot>1$,
since
$
 \costt\prl a \prr \rightarrow \infty
$
as $a\rightarrow \infty$. 
The proof follows by noting that
\begin{align*}
	\kappat \prl 1\prr = G \prl \tinv \prr - \tinv G^{\prime} \prl \tinv \prr + \cmin + D^{\prime} \prl a \prr
	< 0
\end{align*}
if and only if $\cmin + D^{\prime} \prl 1\prr < \Vt$,
where $\Vt$ is defined in \eqref{eq:V}.
\end{proof}

Based on these results, we can characterize the optimal repetition-based strategy.
If $\cmin + D^{\prime} \prl 1 \prr \geq  \Vt$, then $\Not = 1$ since $\aot=1$.
Otherwise, we get $\aot>1$, which is in this case implicitly given by
\begin{align}
	\frac{\partial \costt\prl a\prr}{\partial a}\Bigg\rvert_{a=\aot}
	&= G\prl \frac{1}{\tau \aot }\prr
	- \frac{1}{\tau \aot } G^{\prime} \prl \frac{1}{\tau \aot}\prr
	+ \cmin + D^{\prime} \prl \aot \prr  \nn\\
	&= 0. \label{eq:implicit}
\end{align}
If $1<\aot<2$, then we may get $\Not=1$ or $\Not=2$, based on \eqref{eq:optimalN}.
When $\aot\geq 2$, we get $\Not>1$.

We finally consider the optimal repetition-based strategy
as the target MSE $\tau$ changes.
In the following lemma, we investigate the function $\Vt$
defined in \eqref{eq:V} as $\tau$ changes.

\begin{lemma}
\label{lemma:changingMSE}
$\Vt$ is nonnegative and nonincreasing on $(0,\infty)$,
and in particular, we have $\lim_{\tau\rightarrow \infty} \Vt = 0$, and
\begin{align}
	L \defi \lim_{\tau\rightarrow 0}  \Vt > 0, \label{eq:L}
\end{align}
if $\Vt$ is bounded as $\tau\rightarrow 0$,
or else, the limit does not exist.
\end{lemma}
\begin{proof}
We present the proof in Appendix~\ref{app:changingMSE}.
\end{proof}

It may appear that from \eqref{eq:optimalN} and \eqref{eq:implicit},
as the target MSE $\tau$ decreases, the optimal repetition-based strategy may need
to fuse more units, i.e., $\Not$ may increase.
More rigorously, we next characterize the behavior of the minimizer $\aot$
of the total cost $\costt\prl a \prr$ as the target MSE $\tau$ changes.

\begin{theorem}
\label{thm:changingTau}
If the limit in \eqref{eq:L} exists, and $L \leq \cmin + D^{\prime} \prl 1 \prr $,
then $\aot = 1$ for all $\tau>0$.
If, on the other hand, the limit does not exist, or it exists and $L > \cmin + D^{\prime} \prl 1 \prr $, we define
\begin{align}
	T \defi \inf V^{-1}\prl \cmin + D^{\prime} \prl 1 \prr \prr > 0, \nn
\end{align}
where $V^{-1}\prl x \prr$ is the inverse image of a point $x$ under the function $V$ for any $x>0$.
Then we get $\aot =1$ whenever $\tau \geq T$, and $\aot > 1$ whenever  $0< \tau < T$.
\end{theorem}
\begin{proof}
Suppose the limit in \eqref{eq:L} exists, and $L \leq \cmin + D^{\prime} \prl 1 \prr$.
Then $\Vt \leq \cmin + D^{\prime} \prl 1 \prr$, and $\aot = 1$, for all $\tau>0$.

Suppose next the limit in \eqref{eq:L} either does not exist, or it exists and
$L > \cmin + D^{\prime} \prl 1 \prr$.
Since $V\prl \tau \prr$ is a monotone function, $V^{-1}\prl  \cmin + D^{\prime} \prl 1 \prr  \prr$
is
either a singleton or an interval.
Then for any $\tau \geq T$, we have $\Vt \leq \cmin + D^{\prime} \prl 1 \prr$,
which implies $\aot = 1$,
and when $0<\tau<T$, we have 
$
	\cmin + D^{\prime} \prl 1 \prr < \Vt,
$
which implies $\aot  > 1$.
\end{proof}


In this section, we investigated the cost-performance tradeoff for repetition-based strategies
under convex cost functions.
In particular, we characterized the optimal repetition-based strategy in terms of
the baseline cost, the behaviors of the incremental and fusion cost functions with different parameters,
for different values of the target MSE level $\tau$.
We next study the cost-performance tradeoff under linear cost functions.

\subsection{Linear Cost Functions}
We consider the optimal repetition-based strategy in terms of cost-efficiency
when the underlying cost function is linear, where we can express it as
\begin{align*}
	 C \prl\theta\prr = \cmin + \alpha \theta,\;\;\; \theta>0,
\end{align*}
where $\alpha>0$ is an application-dependent constant.
This case corresponds to a {\it law of proportional returns}.
We show that the strategy $S_1$ is the optimal
repetition-based strategy for any target MSE $\tau>0$.
There is no gain in repetition-based approaches in terms of cost-efficiency for linear cost functions.

\begin{theorem}
\label{thm:linear}
Suppose that the cost function $C\prl \theta \prr$ is linear, that is, $C \prl\theta\prr = \cmin + \alpha \theta$ for some $\alpha>0$.
Then the optimal repetition-based strategy in terms of minimizing the incurred
cost while achieving the same MSE is the strategy $S_1$.
\end{theorem}

\begin{proof}
Let $\tau>0$ be given.
The total cost of the strategy $S_N$, for any positive integer $N$, is given by
\begin{align*}
	\cost_{\tau} \prl N\prr
	= N\cmin +\alpha\sum_{i=1}^{N} \theta_{i} +  D\prl N\prr
	&= N\cmin +\alpha \tinv +  D\prl N\prr\nn\\
	&> \cmin + \alpha\tinv
	= \cost_{\tau} \prl 1\prr.
\end{align*}
This implies the cost incurred by the strategy $S_1$ is smaller than that of
the strategy $S_N$ for any $N>1$ and $\tau>0$.
\end{proof}

For proportional costs a single more reliable unit is always more
cost-efficient than a fusion of several less reliable units
in the sense that it incurs a smaller cost while achieving the same MSE.
We next analyze the concave cost function case.


\subsection{Concave Cost Functions}
We consider the cost-performance tradeoff of each strategy in the class of 
strategies when the cost function is concave.
This case corresponds to a {\it law of increasing returns}, 
as opposed to a law of diminishing returns.
That is, the incremental cost for performance decreases, making single, high-cost, high performance elements more attractive.
Before proving the main theorem of this section,
we present a lemma that proves that the concave incremental cost function
is sub-additive.
\begin{lemma}
\label{lemma:subadditive}
If a function $f$ with the domain $\sql 0,\infty\prr$ is concave, and $f(0)\geq 0$, then it is sub-additive, i.e., for any $x,y\geq 0$,
\begin{align*}
	f(x)+f(y) \geq f(x+y).
\end{align*}
\end{lemma}
\begin{proof}
We provide the proof in Appendix~\ref{app:subadditive}.
\end{proof}

The next theorem characterizes the optimal repetition-based strategy in terms of
minimizing the total incurred cost while achieving the same MSE performance
for a given $\tau>0$.

\begin{theorem}
\label{thm:concave}
Suppose that the cost function $C\prl \theta\prr$ is concave,
and each repetition-based strategy achieves the same MSE level $\tau>0$.
Then the strategy $S_1$ is always the optimal strategy in terms of incurring the smallest cost for any $\tau>0$.
\end{theorem}
\begin{proof}
Let $\tau>0$ be given.
Then, for any positive integer $N$, the total cost incurred by the strategy $S_N$ is given by
\begin{align*}
	\cost_{\tau}\prl N \prr
	&= \sum_{i=1}^N C \prl \theta_{i}\prr + D \prl N\prr
	= \sum_{i=1}^N G \prl \theta_{i}\prr + N\cmin + D \prl N \prr.\nn
\end{align*}
We note that by Lemma~\ref{lemma:subadditive}, the incremental cost function
is sub-additive, since it is concave and $G \prl 0\prr \geq 0$, implying that
\begin{align}
	\sum_{i=1}^N G\prl \theta_{i}\prr 
	\geq G\prl \sum_{i=1}^N \theta_{i}\prr
	= G\prl \tinv \prr.
	\label{eq:concavefirst}
\end{align}
Note that the cost incurred by the strategy $S_1$ is given by
\begin{align*}
	\cost_{\tau}\prl 1 \prr = G\prl \tinv \prr + \cmin,
\end{align*}
implying $\cost_{\tau}\prl N \prr > \cost_{\tau}\prl 1 \prr$ for any $N>1$.
Hence, the strategy $S_1$ is the optimal strategy for any desired MSE.
\end{proof}

Strategy $S_1$, which is formed by exhausting all available cost for a single
computational unit, is more cost-efficient as compared to any strategy $S_N$ with $N>1$,
which allocates available cost across several less reliable computational units.

In this section, we considered the cost-performance tradeoff of repetition-based strategies
under convex, linear, and concave cost function classes.
We showed that under convex cost functions the optimal cost-performance tradeoff may be achieved either by the strategy $S_1$ or by some strategy $S_N$ with $N>1$ under certain conditions.
For linear and concave costs,
optimality is always achieved by strategy $S_1$ for any target MSE performance.
In the next section, we consider applications of our results into a number of
contexts.


\section{Applications}
\label{sec:app}
Here, we show how our cost-fidelity formulation and theoretical results are connected
to problems from different fields.

\subsection{Neuroscience}
We review a particular application of our framework in a theoretical neuroscience context.
We focus on two principal tasks of the brain where synapses play essential roles, namely, information storage and information processing.
Typical central synapses exhibit noisy behavior due, for instance, to probabilistic transmitter release.
The firing of the presynaptic neuron is inherently stochastic and occasionally fails to evoke an excitatory postsynaptic potential (EPSP).
In this sense, we can cast each noisy synapse as an unreliable computational unit,
contributing to the overall neural computation carried out by its efferent neuron.
We focus on two distinct cost-fidelity formulations, where we show
that experimental results \cite{VarshneySC2006, MishchenkoHSMHC2010} agree with our theoretical predictions.
We note that recall corresponds to a form of ``in-memory computing'' whereas processing corresponds to a form of ``in-sensor computing''.

\subsubsection{In-Memory Computing}
Revisiting \cite{VarshneySC2006}, we first consider an information-theoretic framework to study the information storage capacity of synapses under resource constraints, where memory is seen as a communication channel subject to several sources of noise.
Each synapse has a certain SNR, where increasing the SNR increases the information storage capacity in a logarithmic fashion.
However, this increase comes at a cost, namely, the synaptic volume.
Hence, from an information storage perspective, we cast capacity as the fidelity of a noisy synapse and the volume as the cost.
If we denote the information storage capacity of a synapse and its average volume by $C_{\mathrm{I}}$  and $V$, respectively, then taking Shannon's AWGN channel capacity formula \cite{Shannon1948} for concreteness:
\begin{align*}
     C_{\mathrm{I}} = \frac{1}{2} \ln\prl 1 + \frac{V}{V_N} \prr,
\end{align*}
where $V_N$ is the volume of a synapse with a unit SNR.
This relationship assumes the power law
$(V/V_N) = \prl A/A_N\prr^2$, which is supported by experimental measurements  \cite{VarshneySC2006},
where $A$ is the mean EPSP amplitude and $A_N$ is the noise amplitude.
We rewrite the volume as a function of capacity as
\begin{align*}
     V = V_N \prl e^{2C_{\mathrm{I}}} -1 \prr,
\end{align*}
and observe that this is an exponential cost function, a particular example of convex costs.
For exponential costs, fusion of several less reliable computational units may lead to better cost-efficiency than a single more reliable computational unit.
Therefore, our cost-fidelity framework applied to information recall under resource constraints recovers the principle that several small and noisy synapses should be present in brain regions performing storage and recall, rather than large and isolated synapses \cite{VarshneySC2006, BrunelHINB2004}.

Moreover, \cite{LaughlinRSA1998, Zador1998, ManwaniK2001, LevyB2002, Goldman2004} show that the noisiness of the synapses leads to efficient information transmission.
That is, transmitting the same information over several less reliable but metabolically cheaper synapses requires less energy, as compared to the case where the information is transmitted over a single, more reliable but metabolically more expensive synapse.
The idea that noise can facilitate information transmission is also present in neuronal networks.
In particular, the authors in \cite{LaughlinS2003} show that a neuron is a noise-limited device of restricted bandwidth, and an energy-efficient nervous system will split the information and transmit it over a large number of relatively noisy neurons of lower information capacity.

\begin{figure}[t]
	\centering
	\includegraphics[scale=0.47]{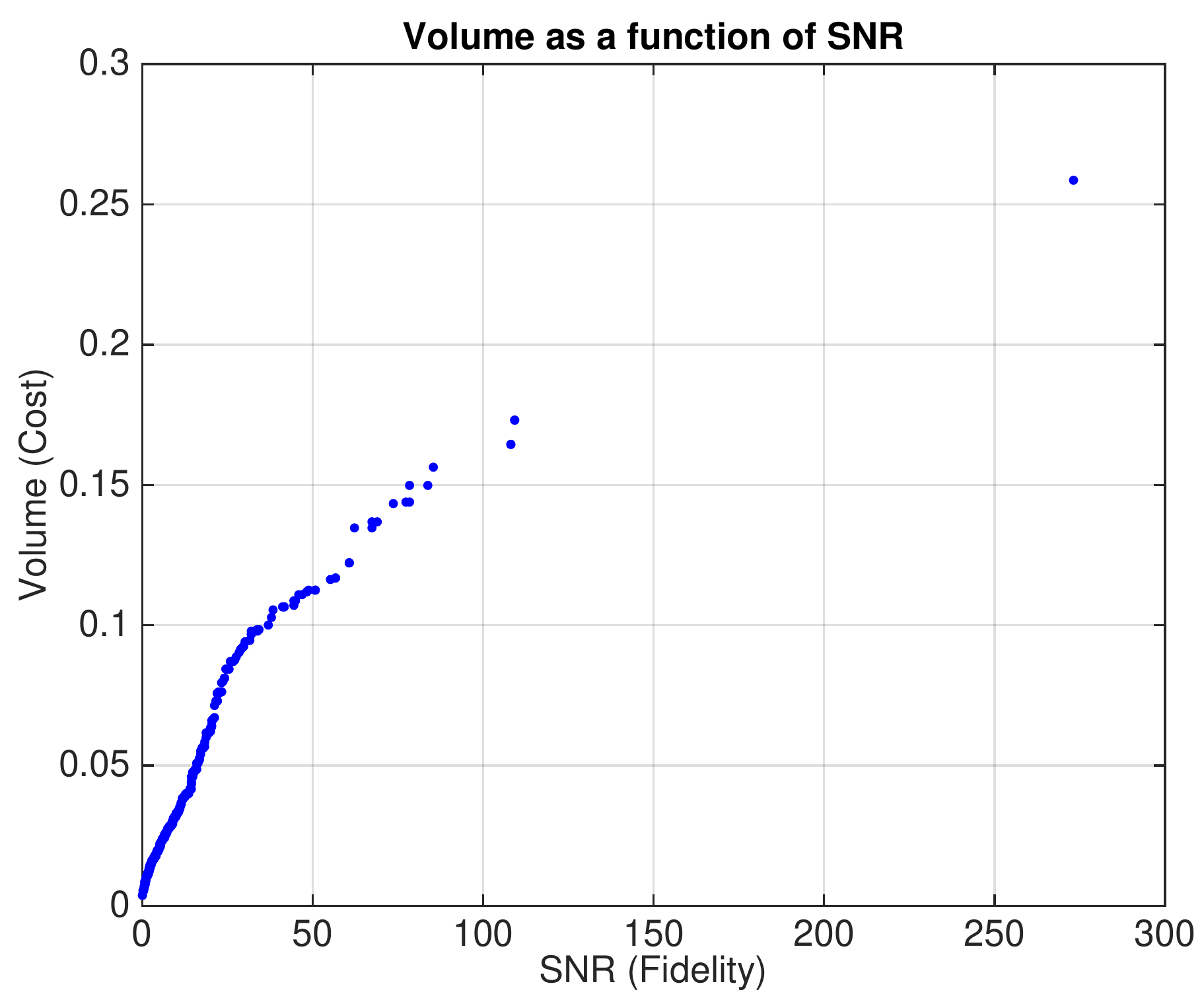}
	\caption{A data-driven cost (volume in $\mu m^3$) versus fidelity (SNR) function.}
	\label{fig:synapses}
\end{figure}

\subsubsection{In-Sensor Computing}
We next consider an information processing perspective,
and view the SNR of a synapse itself as its fidelity and the synaptic volume as the cost.
We adopt a data-driven approach using two different data sets.
This joining is necessary since joint electrophysiology and imaging experiments are technically difficult,
where electrophysiology experiments to measure voltages require live tissue
while electron micrograph imaging experiments to measure volumes require fixing and slicing the tissue
\cite{MishchenkoHSMHC2010}.

The first data set \cite{VarshneySC2006} includes EPSP measurements across 637 distinct
synapses over 43 trials for each synapse.
Based on these measurements, we generate an empirical distribution of the mean EPSP measurements of a synapse.
The second data set \cite{MishchenkoHSMHC2010} includes volume measurements across
357 synapses, which is used to compute a distribution of a synapse volume.

We first generate $T=500$ random variables $\crl Y_t \crr_{t=1}^T$
from the calculated volume distribution.
We next generate $T$ random variables from the calculated mean EPSP distribution,
and sort them assuming a monotonic relationship between the mean EPSP and the volume of synapses \cite{VarshneySC2006}.
From the sorted mean EPSP amplitudes, we compute the corresponding SNRs
$\crl X_t \crr_{t=1}^T$.
We plot the resulting pairs $\crl \prl X_t, Y_t \prr\crr_{t=1}^T$ in Figure~\ref{fig:synapses}.
This plot indicates that the cost function is approximately concave as a function of SNR.
More rigorously, we assess convexity using a nonparametric hypothesis test based on a simplex statistic,
a descriptive measure of curvature described in \cite{AbrevayaJ2005}.
When applied to this data, the test yields a $p$-value of $3.25\times10^{-4}$, which can be interpreted as a strong evidence in favor of the hypothesis that
the cost (volume) is a concave function of the SNR (fidelity).
This suggests that the brain may achieve cost-efficiency by using a single large and
reliable synapse, instead of several smaller and less reliable synapses, from an information processing perspective.

To compare this prediction with experimental findings, we focus on a particular synapse
called the {\it calyx of Held}, the largest synapse in the mammalian auditory central nervous system
that connects principal neurons within the auditory system \cite{Morest1968, SmithJCY1991, Hermann2008}.
The calyx of Held plays a crucial role in certain information processing
tasks of the brain.
For instance, the principal cells connected by the calyx of Held enable interaural level detection, 
a vital role in high frequency sound localization \cite{SpanglerWH1985, Tsuchitani1997}.
The signals derived from the calyx of Held generate large excitatory postsynaptic currents
with a short synaptic delay, where the transmission speed and fidelity of the calyx is 
very reliable in mature animals \cite{FutaiOMT2001}.

Hence, the calyx of Held may be regarded as a very reliable but costly synapse, as compared to the
ones performing information storage tasks, which are noisier and less costly in
terms of brain resources.
We observe that these experimental findings agree with our prediction that the cost-efficiency
results from employing a single reliable and costly synapse (calyx of Held),
instead of several less reliable and metabolically cheaper synapses, under a concave cost function.

\begin{figure}[t]
	\centering
	\includegraphics[scale=0.5]{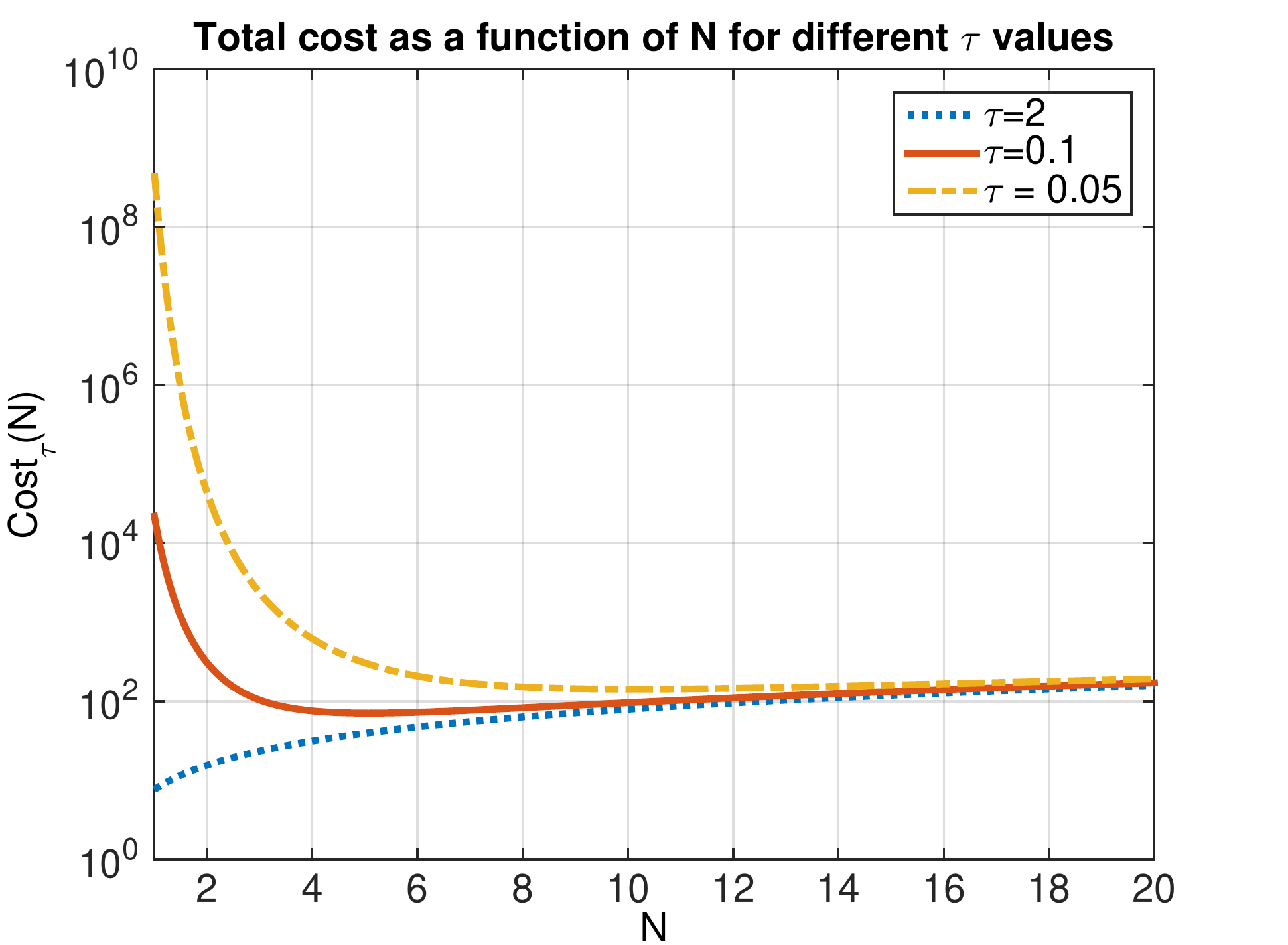}
	\caption{Total cost function \eqref{eq:exampleCost} with
	$\alpha=1,\beta=1,\gamma=1,\cmin=7$ for different values of the target MSE level $\tau$.}
	\label{fig:costVSN}
\end{figure}

\subsection{Circuits}
Next, let us consider signal processing systems implemented on unreliable circuit fabrics.
As CMOS technology scales beyond $\SI{10}{\nm}$, the operation of CMOS devices begins to suffer from static defects as well as dynamic operational non-determinism \cite{WongFSWW1999, WangS2003, ChoiSSC2007}.
Moreover, spintronics, which use electron spin for computing, exhibit an unreliable behavior,
where there is a tradeoff between reliability and energy consumption
\cite{ButlerMMVPRH2012, PatilMNYS2016}.
That is, probability of failure is smaller when more energy is used.
Hence, deeply scaled CMOS and spintronics based systems must operate
under computational errors.

In \cite{Neumann1956}, von Neumann studied noise in circuits and showed that even when circuit
components are unreliable, reliable computations can be performed by using repetition-based schemes.
Repeated computations followed by a majority vote have also been used extensively in error-tolerant circuit design \cite{AbdallahS2013, HanLL2014}.
Also, Hadjicostis \cite{Hadjicostis1999} investigated redundancy-based approaches
to build fault-tolerant dynamical systems out of cheap but unreliable components.

Moreover, a statistical error compensation technique called Algorithmic Noise Tolerance (ANT)
has been studied in \cite{HegdeS2001, KimS2013}.
ANT compensates for errors in computation in a statistical manner
by fusing outcomes of several unreliable computational branches that operate at different points
along energy-reliability tradeoffs.
The ANT framework can also be cast as a CEO problem in multiterminal source coding \cite{SeoV2016}.

Stochastic behavior in circuit fabrics may arise when computation is
embedded into either memory, which leads to in-memory computing \cite{KangKSEC2014},
or sensing, which leads to in-sensor computing \cite{HuRSSSWV2012},
to achieve cost-efficiency \cite{Shanbhag2016}.
Note that in-memory computing and in-sensor computing may lead to fundamentally different
cost-performance tradeoffs.
In particular, we demonstrate that the difference between in-memory computing and in-sensor computing may be modeled through our framework by using different cost-fidelity function classes.

\subsubsection{Example case}
Here, we present an application of the results of this section into spintronics.
In particular, exponential cost has been shown to approximately model the functional dependence between energy and reliability for a typical spin device \cite{PatilMNYS2016}.
Consider the exponential cost
\begin{align}
	C\prl \theta \prr = \cmin + \alpha \prl e^{\beta\theta} - 1\prr,\;\; \theta>0,
	\label{eq:exponentialExample}
\end{align}
for some $\alpha,\beta>0$.
Moreover, for illustration purposes, we assume that the fusion cost function is
$
	D\prl N \prr =\gamma \prl N- 1\prr,
$
for $N\geq 1$ and $\gamma>0$.
Then the total cost function is given by
\begin{align}
	\costt \prl N \prr
	&= \alpha N \prl e^{ \frac{\beta}{\tau N}} - 1 \prr + N\prl \cmin + \gamma \prr - \gamma,
	\label{eq:exampleCost}
\end{align}
for any positive integer $N$.
In Fig.~\ref{fig:costVSN}, we plot this total cost function with parameters $\alpha=1,\beta=1,\gamma=1,\cmin=7$ for different values of the target MSE $\tau>0$.
We observe that Fig.~\ref{fig:costVSN} illustrates how $\Not$ increases as $\tau$ decreases, as discussed in this section.
In particular, we note that $\Not=1,6,13$ for $\tau=2,0.1,0.05$, respectively.

Finally, the total cost function \eqref{eq:exampleCost} yields
\begin{align}
	V\prl \tau \prr
	&= \alpha \exp\prl \beta \tinv \prr \prl \beta\tinv - 1 \prr + \alpha,
	\label{eq:exampleV}
\end{align}
implying $V\prl \tau \prr \rightarrow \infty$ as $\tau \rightarrow 0$.
Hence there exists a threshold
\begin{align*}
	T = V^{-1} \prl \cmin + \gamma \prr > 0
\end{align*}
such that $\aot=1$ when $\tau \geq T$, and $\aot>1$ when $\tau < T$.
These cases are illustrated in Fig.~\ref{fig:Vtau} for $\cmin=7,\gamma=1$.

\begin{figure}[t]
	\centering
	\includegraphics[scale=0.5]{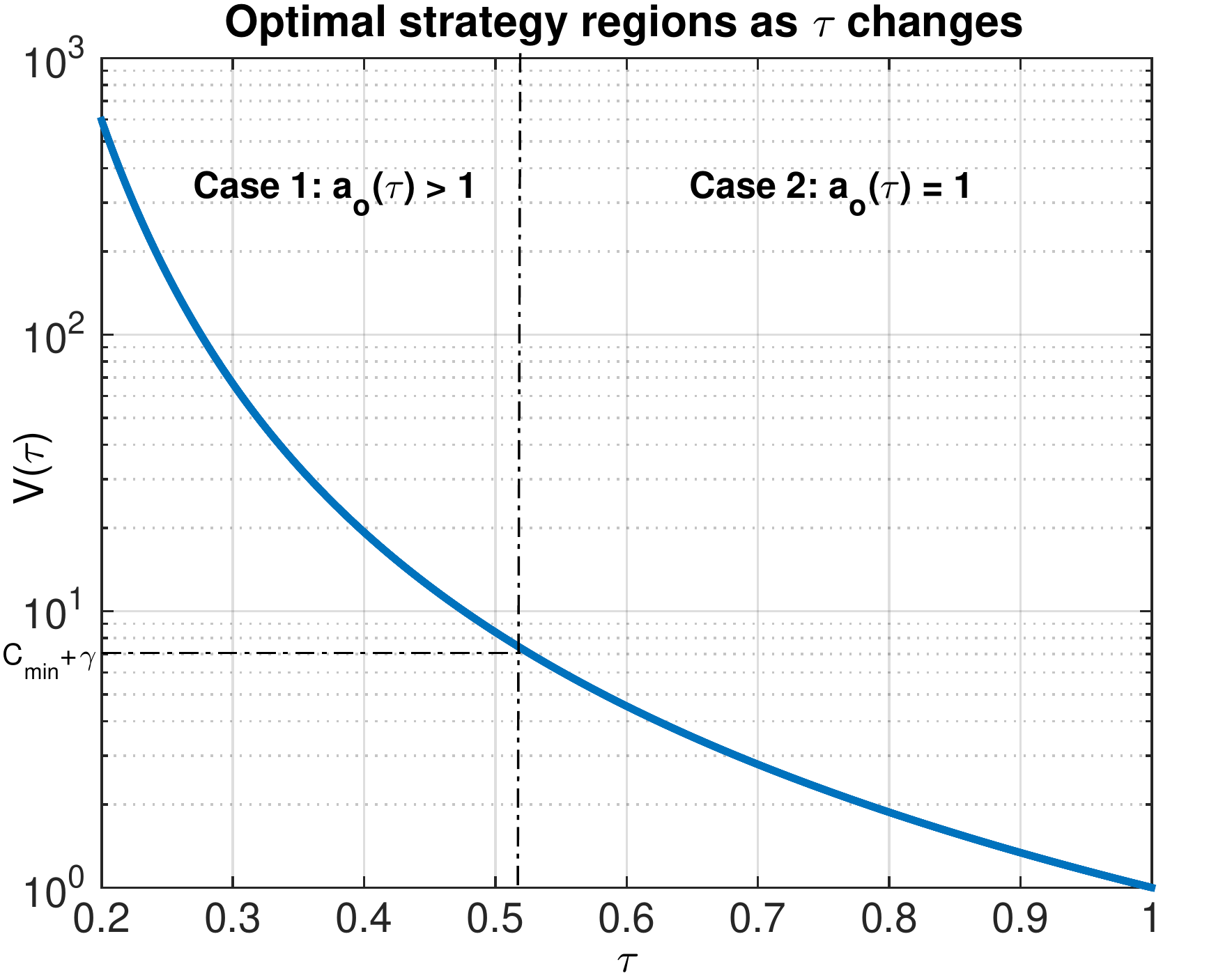}
	\caption{The function $V\prl \tau \prr$ \eqref{eq:exampleV} to illustrate the optimal strategy regions.}
	\label{fig:Vtau}
\end{figure}

\subsection{Crowdsourcing}
Crowdsourcing assigns a task to a large number of less expensive but unreliable workers,
instead of a small number of more expensive and reliable experts.
Monetary payment to incentivize workers has been shown to
affect the quality and the quantity of work in such scenarios \cite{JianhanPV2016}.
Recently, motivated by reliability issues of crowdsourced workers and limited budgets,
several researchers have pursued the limits of achievable performance
from
estimation-theoretic \cite{JianhanPV2016},
information-theoretic \cite{LahoutiH2016},
optimization \cite{KargerOS2014, KhetanO2016},
and empirical \cite{HoSSV2015}
perspectives.

The authors of \cite{HoSSV2015} studied the relation between monetary incentives and work quality in a knowledge task.
More precisely, they performed an experiment on 451 unique workers on Amazon Mechanical Turk, and investigated the effect of bonus payments on the work quality in the task of proofreading an article.
They measured the quality by the number of typographical errors found in a given article.
In this scenario, each worker is paid a base salary (minimum cost),
and an additional bonus (incremental cost), which is shown to yield an improvement in the work quality.
In this sense, the bonus payment, i.e., the incremental cost, can be viewed as a function of
the number of errors found.
In particular, experiments in \cite{HoSSV2015} showed that increasing the bonus payment has diminishing returns in terms of the work quality.
That is, the incremental cost is a convex function of the work quality.

More recently, Lahouti and Hassibi \cite{LahoutiH2016} considered the crowdsourcing problem as a
human-based computation problem where the main task is inference.
They formulated an information-theoretic framework, where
unreliable workers are modeled as parallel noisy communication channels.
They represented the queries of the workers and the final inference using a joint source channel encoding/decoding scheme.
Similarly, Khetan and Oh \cite{KhetanO2016} studied the tradeoff between budget and accuracy in crowdsourcing scenarios under the generalized Dawid-Skene model, where
they introduced an adaptive scheme to allocate a budget across
unreliable workers.

We observe that there is a tradeoff between cost (monetary payments, bonus) and fidelity (quality of
work) in a wide range of crowdsourcing scenarios.
In particular, assigning a task to several workers, distributing the limited budget among them, and fusing their unreliable outputs have been problems of interest in the crowdsourcing literature.
In this sense, our cost-fidelity formulation and repetition-based approaches may have relevance in crowdsourcing problems.

\section{Conclusion and Future Directions}
\label{sec:conc}
We considered fusing outcomes of several unreliable
computational units that perform the same task.
We modeled unreliability in a computational outcome using an additive perturbation,
where the fidelity is inversely related to the variance of the perturbation.
We investigated cost-performance tradeoffs achievable through repetition-based approaches.
Here, each computational unit incurs a baseline cost as well as
an incremental cost, which is a function of its fidelity.

We defined a class of repetition-based strategies, where any repetition-based strategy
distributes the cost across several unreliable computational units and fuses their outcomes
to produce a final output, where it incurs cost to perform the fusion operation.
We considered the MSE of each strategy in estimating the error-free computation.
In particular, we defined the optimal repetition-based strategy as
the one incurring the smallest cost while achieving the desired MSE performance.

When the cost is a convex function of fidelity,
the optimal repetition-based strategy may distribute cost across several less reliable
computational units instead of using a single more reliable unit under certain conditions.
For the classes of concave and linear cost functions we preserved
that the optimal strategy uses only a single and relatively reliable computational unit,
instead of a fusion of several less costly but less reliable units.

We assumed that outcomes produced by different computational units
are uncorrelated.
This framework can be extended to a correlated outcome setting.
When studying the fundamental tradeoff between cost and performance,
we assumed that the fusion operation is error-free.
We can extend this to the case where
the fusion operation also produces noisy results under cost and fidelity constraints.
Moreover, we focused on a particular fusion operation, i.e.,
linear combination, which is common in certain applications.
More generally, we can consider nonlinear fusion rules to compute
the final estimate of the error-free computation.
For instance, midrange \cite{JianhanPV2016} and median-of-means \cite{DevroyeLLO2015} estimators have been considered as alternatives to linear estimators under different scenarios to improve performance.
Extension of this setup would be of interest for different network topologies,
as opposed to the centralized fusion setting of this paper, as in \cite{XuR2016}.
\section*{Acknowledgements}
We thank Dmitri B. Chklovskii for providing data from \cite{MishchenkoHSMHC2010}.

\appendix
\numberwithin{equation}{section}

\section{Proof of Lemma~\ref{lemma:compare}}
\label{app:compare}
The MSE of the strategy $S_N$ with a given
$\vtheta\in\prl 0,\infty\prr^N$ is
\begin{align}
	\MSE \prl \vw, \vtheta \prr
	&= \E\sql \prl Y \prl \vw^T \vOne - 1\prr + \vw^T \vU \prr^2\sqr, \nn
\end{align}
where \eqref{eq:cont_single} is substituted in \eqref{eq:MSEdefined}.
Since $Y$ and $\vU$ are uncorrelated:
\begin{align}
	\MSE \prl \vw, \vtheta \prr
	=\E\sql Y^2 \sqr \prl \vw^T \vOne - 1\prr^2 + \vw^T \Sigma_{\vU} \vw,\label{eq:MSEunopt}
\end{align}
where
$\Sigma_{\vU}$ is the covariance matrix of the perturbation vector $\vU$.
If we impose the condition that $\vw^T\vOne = 1$ in \eqref{eq:MSEunopt}, then
\begin{align}
	\MSE \prl \vw, \vtheta \prr
	= \vw^T \Sigma_{\vU} \vw. \nn
\end{align}
To minimize this over weights that satisfy $\vw^T\vOne = 1$,
we first form the Lagrangian
\begin{align*}
	J\prl \vw, \lambda \prr = \frac{1}{2} \vw^T \Sigma_{\vU} \vw + \lambda \prl 1 - \vw^T \vOne \prr,
\end{align*}
and then compute the gradient with respect to $\vw$ to get
\begin{align*}
	\Sigma_{\vU} \vw - \lambda \vOne = 0,
\end{align*}
which is satisfied iff $\vw = \lambda \Sigma_{\vU}^{-1} \vOne$.
With $\vw^T\vOne = 1$, it yields
\begin{align*}
	\lambda = \frac{1}{\vOne^T \Sigma_{\vU}^{-1} \vOne},
\end{align*}
which yields the optimal weights
\begin{align*}
	\vwo = \frac{1}{\vOne^T \Sigma_{\vU}^{-1} \vOne} \Sigma_{\vU}^{-1} \vOne.
\end{align*}
Finally, when substituted in $\MSE\prl \vw,\vtheta \prr$, we achieve
\begin{align}
	\MSEo \prl \vtheta \prr
	= \vwo^T \Sigma_{\vU} \vwo
	= \frac{1}{\vOne^T \Sigma_{\vU}^{-1} \vOne}.\nn
\end{align}
The proof follows by noting
$\Sigma_{\vU} = \mathrm{diag}\prl \theta_1^{-1},\ldots,\theta_N^{-1} \prr$.

\section{Proof of Lemma~\ref{lemma:uniform}}
\label{app:uniform}
We solve this optimization problem using the method of Lagrange multipliers, where we first form
the Lagrangian
\begin{align*}
	&J\prl \theta_{1},\ldots,\theta_{N},\lambda \prr
	\triangleq \sum_{i=1}^{N} C\prl \theta_{i}\prr 
	+ \lambda \prl \tinv - \sum_{i=1}^N\theta_{i} \prr.
\end{align*}
Then, we set the derivative of the Lagrangian with respect to $\theta_{j}$ to $0$, which is given by
\begin{align*}
	&\frac{\partial J}{\partial \theta_{j}}\prl \theta_{1},\ldots,\theta_{N},\lambda \prr\nn
	=C^{\prime}\prl \theta_{j}\prr -\lambda =0,
\end{align*}
for each $j=1,\ldots,N$. Hence the necessary conditions for optimality are given by
$
	\lambda =C^{\prime}\prl \theta_{j}\prr
$
for $j=1,\ldots,N$.

Here, we note that the cost function $C\prl \theta\prr$ is convex
and strictly increasing in $\theta$, and
its derivative $C^{\prime}\prl \theta\prr$ is nondecreasing.
This implies that it is invertible, so we can write
\begin{align*}
	\theta_{j} = \prl C^{\prime} \prr^{-1} \prl \lambda\prr,
\end{align*}
for each $j=1,\ldots,N$, where $\prl C^{\prime} \prr^{-1}$ is the inverse of the function
$C^{\prime}$.
That is, $\theta_{1}=\cdots=\theta_{N}$.
Moreover, by imposing the MSE constraint, we get
$
	\theta_j = \prl \tau N \prr^{-1}
$
for any $j=1,\ldots,N$, which yields the desired result.

\section{Proof of Lemma~\ref{lemma:convexTotal}}
\label{app:convexTotal}
We first differentiate the total cost function as
\begin{align*}
	\frac{ \partial \costt \prl a\prr}{\partial a }
	= G\prl \frac{1}{\tau a}\prr - \frac{1}{\tau a} G^{\prime} \prl \frac{1}{\tau a}\prr
	+ \cmin + D^{\prime} \prl a \prr.
\end{align*}
We next find its second derivative as
\begin{align*}
	&\frac{ \partial^2 \costt \prl a\prr}{\partial a^2 }\nn\\
	&= - \frac{1}{\tau a^2} G^{\prime} \prl \frac{1}{\tau a}\prr
	+ \frac{1}{\tau a^2} G^{\prime} \prl \frac{1}{\tau a}\prr
	+ \frac{1}{\tau^2 a^3} G^{\prime\prime} \prl \frac{1}{\tau a}\prr
	+ D^{\prime\prime} \prl a \prr \nn\\
	& = \frac{1}{\tau^2 a^3} G^{\prime\prime} \prl \frac{1}{\tau a}\prr
	+ D^{\prime\prime} \prl a \prr,
\end{align*}
which is nonnegative since the incremental cost function $G\prl \cdot \prr$ and the fusion cost function
$D\prl \cdot\prr$ are both convex and $a>0$.

\section{Proof of Lemma~\ref{lemma:changingMSE}}
\label{app:changingMSE}
We first observe that from \eqref{eq:V}
\begin{align*}
	V^{\prime} \prl \tau \prr
	& = -\frac{1}{\tau^2} G^{\prime} \prl \tau^{-1}\prr - \frac{1}{\tau^3} G^{\prime\prime}\prl \tau^{-1} \prr
	+ \frac{1}{\tau^2} G^{\prime} \prl \tau^{-1}\prr \nn\\
	& =- \frac{1}{\tau^3} G^{\prime\prime}\prl \tau^{-1} \prr \leq 0,
\end{align*}
for any $\tau>0$, as $G\prl \cdot \prr$ is convex and twice differentible.
Thus, the function $V\prl \tau \prr$ is decreasing on $\prl 0,\infty\prr$.
We next note that
\begin{align*}
	\lim_{\tau \rightarrow \infty} V\prl \tau \prr
	&= \lim_{\tau \rightarrow \infty} \prl \tinv G^{\prime}\prl \tinv \prr - G\prl \tinv \prr \prr \nn\\
	&= \lim_{\tau \rightarrow \infty} \tinv G^{\prime}\prl \tinv \prr - G\prl 0 \prr 
	= 0,
\end{align*}
since $G\prl 0 \prr = 0$ and $G^{\prime}\prl 0 \prr$ is finite.
Therefore, $V\prl \tau \prr$ is nonnegative on $\prl 0,\infty\prr$.
This implies that the function $V\prl \tau \prr$
either converges to a finite limit (if and only if $V\prl \tau \prr$ is bounded on $\prl 0,\infty \prr$),
or is unbounded as $\tau \rightarrow 0$.

\section{Proof of Lemma~\ref{lemma:subadditive}}
\label{app:subadditive}
Suppose that $\lambda\in[0,1]$.
Since $f$ is concave, we have
\begin{align*}
	f\prl \lambda x\prr &= f\prl \lambda x + (1-\lambda)0\prr\\
					&\geq \lambda f(x)+(1-\lambda)f(0)\geq \lambda f(x).
\end{align*}
Then, for any $x,y>0$, we can write
\begin{align*}
	f(x)+f(y) &= f\prl (x+y)\frac{x}{x+y}\prr+f\prl (x+y)\frac{y}{x+y}\prr.\nn\\
	              & \geq \frac{x}{x+y}f\prl x+y \prr + \frac{y}{x+y}f\prl x+y \prr \\ 
	              &= f(x+y),
\end{align*}
where we use $x/(x+y),y/(x+y)\in[0,1]$.

\bibliographystyle{unsrt} 
\bibliography{abrv,conf_abrv,mad_lib}
\end{document}